\documentclass[final,leqno]{article}
\usepackage[T1]{fontenc} 
\usepackage{ae,aecompl}	 
\usepackage{times}
\usepackage{fancyhdr}
\usepackage{amsfonts}\usepackage{amsmath}\usepackage{amssymb}
\usepackage[hang,small,bf]{caption}
\usepackage[latin1]{inputenc} 
\usepackage{hyperref}\usepackage{breakurl} 
\usepackage{graphicx} 
\usepackage{float}
\usepackage{subfig}
\usepackage{cleveref} 
\usepackage{soul} 
\usepackage{comment}

\newtheorem{theorem}{Theorem}

\newenvironment{proof}[1][Proof]{\textbf{#1.} }{\ \rule{0.5em}{0.5em}}

\begin{document}

\title{Hospitalization in the transmission of dengue dynamics: The impact on public health policies\\[2cm]
 Hospitalizaciones en la din\'amica de transmisi\'on del dengue: El impacto en pol\'iticas de salud p\'ublica  }
\author{
Fabio Sanchez\thanks{CIMPA-Escuela de Matem\'atica, Universidad de Costa Rica, San Jos\'e, Costa Rica. 
E-Mail: \href{mailto: fabio.sanchez@ucr.ac.cr}{fabio.sanchez@ucr.ac.cr}}
\and 
{\sc Jorge Arroyo-Esquivel}\thanks{Department of Mathematics, University of California Davis, California, USA.  
E-Mail: \href{mailto: jarroyoe@ucdavis.edu}{jarroyoe@ucdavis.edu}}\and
{\sc Paola V\'asquez}\thanks{Escuela de Salud P\'ublica, Universidad de Costa Rica, San Jos\'e, Costa Rica.  
E-Mail: \href{mailto: paola.vasquez@ucr.ac.cr}{paola.vasquez@ucr.ac.cr}}
} 


\maketitle

\begin{abstract}
Dengue virus has caused major problems for public health officials for decades in tropical and subtropical countries. We construct a compartmental model that includes the risk of hospitalization and its impact on public health policies. The {\it basic reproductive number}, $\mathcal{R}_0$, is computed, as well as a sensitivity analysis on $\mathcal{R}_0$ parameters and discuss the relevance in public health policies. The local and global stability of the disease-free equilibrium is established. Numerical simulations are performed to better determine future prevention/control strategies.
\end{abstract}


\section{Introduction}
Dengue is a vector-borne viral infection that inflicts substantial health, economic, and social burden to more than 100 countries in tropical and subtropical regions around the world \cite{brady2012refining}. Globalization, climate change, unplanned urbanization, and insufficient mosquito control programs \cite{gubler2002epidemic,murray2013epidemiology}, are among the complex factors that have contributed to the geographic expansion and rise in the global incidence of a disease, that is now causing an estimated 390 million infections annually, of which an approximated 96 million have clinical manifestations \cite{bhatt2013}.

The highly anthropophilic \textit{Aedes aegyti} mosquito, is the predominant vector \cite{ponlawat2005} of the four antigenically-distinct, but closely related, dengue viruses (DENV-1, DENV-2, DENV-3, and DENV-4) \cite{gubler1995dengue}, while \textit{Aedes albopictus} is considered secondary and less efficient in urban settings \cite{rezza2012}. Each one of these serotypes generates a unique host immune response to the infection \cite{simmons2012}, where infection with a particular serotype will confer lifelong immunity to that strain, but only short-term cross-immunity between serotypes \cite{reich2013}. In areas where dengue is endemic, different serotypes often co-circulate, therefore,  infections with heterologous DENV serotypes are common and can lead to an increased risk of more severe clinical manifestations \cite{ohainle2011}. However, the broad clinical spectrum of dengue infections \cite{world2009dengue} depends not only in the individual DENV infection history, but in a variety of factors like age of the human host \cite{kittigul2007}, host genetic susceptibility \cite{Carvalho2017}, possible chronic diseases \cite{world2009dengue} as well as, in the specific DENV serotype and genotype causing the infection \cite{Yung2015}. 

After the bite of an infected mosquito, a large proportion of the individuals will be asymptomatic, and a silent reservoir of crucial importance in the dynamics of dengue spread \cite{duong2015}. Individuals that do experience symptoms the most common outcome is a self-limiting febrile illness that does not progress to more severe forms, and do not need admission to a hospital. However, many of them will be unable to develop their every day activities, resulting in a major economic burden \cite{castro2017}. Those that evolve, and require hospitalization show a variety of warning signs as constant and intense abdominal pain, persisting vomiting, ascitis, pleural o pericardial effusion, mucosal bleeding, lethargy, lipothymia, hepatomegaly, and progressive increase in the hematocrit \cite{world2012dengue}. The disease can also progress to hemorrhagic and shock complications that can lead to death of the patient, although rare \cite{world2009dengue}. It is estimated that each year, approximately 500,000 cases require hospitalization \cite{denguebulletin}, with the number of deaths globally increasing by 65.5\% between 2007 and 2017 (from 24,500 to 40,500 deaths) \cite{roth2018global}. 


In Costa Rica, dengue has been a significant public health challenge since 1993, when after more than 30 years of absence, autochthonous cases were reported on the Pacific coast \cite{morice2010}. Since then, and despite efforts made, dengue infections have been documented annually, with peaks of transmission observed seasonally (within the year) and cyclical every 2-5 years. A total of 376,158 clinically suspected and confirmed cases \cite{mscr} have been reported to the Ministry of Health, of which  more than 45,000 cases have required hospital care \cite{mscr,ccss}. DENV-1, DENV-2 and DENV-3 have circulated the country in different moments and 1,196 of the total of cases, have been classified as severe dengue, which lead to 23 deaths \cite{mscr}.
 
Given the complexity involved in the transmission of vector borne diseases, such as dengue, mathematical models play a significant role to better understand the interactions between the vector, the virus and the host. In this article we develop a compartmental model that analyses the effect of discriminating the hospitalized (diagnosed) infected individuals and its effectiveness on the overall behavior of the dynamics of dengue fever, which can provide information for public health authorities to implement better prevention and control approaches. 

The article is divided in the following sections. In Section \ref{sec:model}, we introduce a compartmental model that describes the transmission dynamics of dengue fever; Section \ref{sec:results}, presents the results of the model; in Section \ref{sec:disc}, we give a discussion and concluding remarks.

\section{Model with risk of hospitalization} \label{sec:model}
We introduce a model that describes the dynamics of dengue between hosts and vectors in Costa Rica that includes the role of hospitalizations. According to data provided by the Costa Rican Social Security Fund (CCSS), during the last two decades, a total of 45,577 patients with DENV infections, have required hospital care services \cite{ccss}, which represents 13\% of the total of reported confirmed and suspected cases reported during that same period. In Figure \ref{fig:map}, we illustrate the concentration of hospitalizations due to the DENV in the country. As seen in the map, the vast majority of hospitalized patients were reported from regions near the coasts, where temperature is ideal for mosquito prevalence and with circulation of the other two arboviruses, zika and chikungunya. Lim\'on, a province located in the Caribbean coast, reported a total of 17,894 hospitalized cases, follow by Guanacaste with 12,233 cases and Puntarenas with 8,244 cases, both of them located in the Pacific coast. Patients in the age group of 20-44 years of age represented the 43.8\% of the total of hospitalized patients and no significant difference between men and women was observed \cite{ccss}. 

\begin{figure} [htb]
\centering
\subfloat[Dengue hospitalizations by region, Costa Rica 1997-2017.]{%
\resizebox*{5cm}{!}{\includegraphics{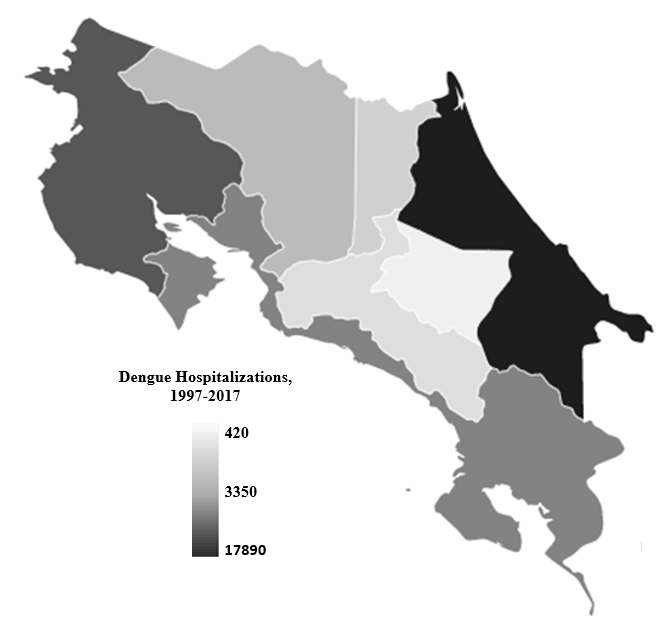}}}\hspace{5pt}
\subfloat[Dengue hospitalizations by sex and age group, Costa Rica 1997-2017.]{%
\resizebox*{7cm}{!}{\includegraphics{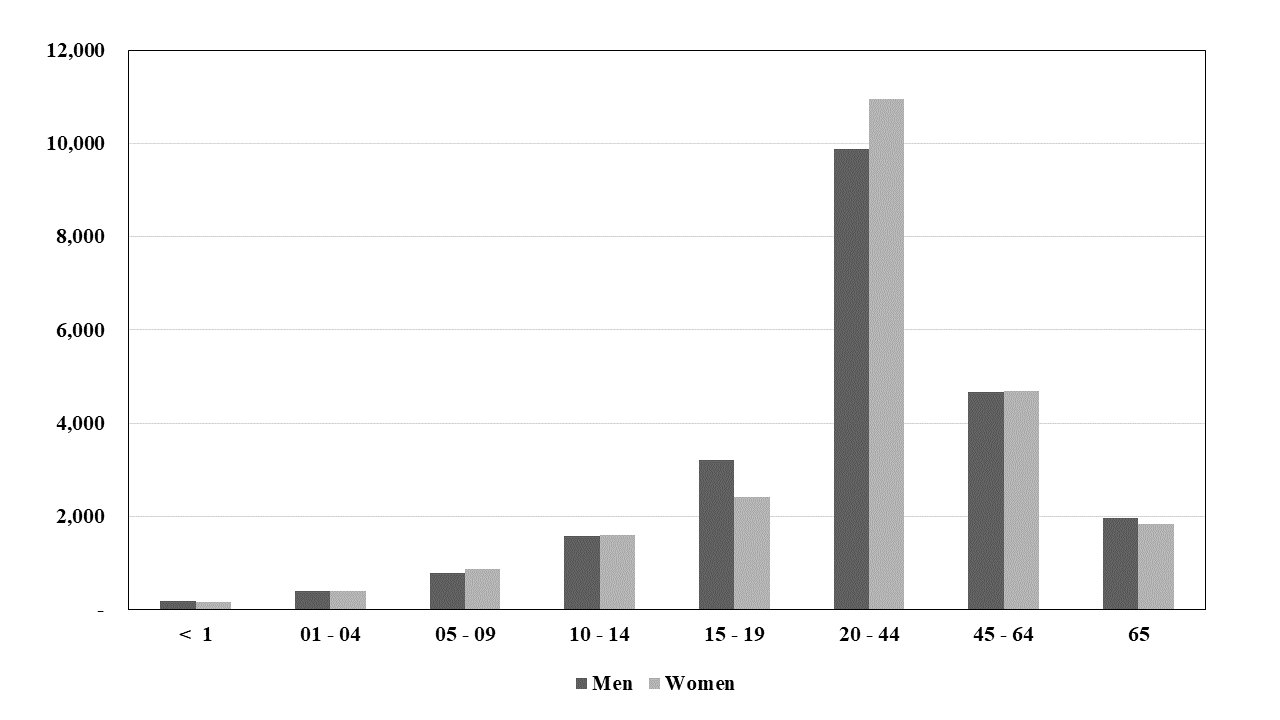}}}
\caption{In Costa Rica, the vast majority of hospitalizations were reported in the coastal regions. These areas, are characterized by been highly infested by the \textit{Aedes} mosquito and by having the higher incidence of dengue cases throughout the year \cite{ccss}.} 
\label{fig:map}
\end{figure}

Model state variables are presented in \Cref{tb:vars}. The following is the system of nonlinear differential equations:

\begin{equation}
\begin{array}{rcl}
S_h^\prime &=& \mu_h N_h - \beta_{hv} S_h \frac{I_v}{N_v}-\mu_h S_h,\\
E_h^\prime &=&  \beta_{hv} S_h \frac{I_v}{N_v} - (\mu_h+\alpha_h)E_h,\\
I_h^\prime &=& (1-p)\alpha_h E_h - (\mu_h+\gamma+\delta)I_h,\\ 
H^\prime &=& p\alpha_h E_h + \delta I_h - (\mu_h+\gamma)H,\\
R^\prime &=& \gamma I_h + \gamma H - \mu_h R,\\
S_v^\prime &=& \mu_v N_v- \beta_{vh} S_v \frac{(I_h+(1-\eta)H)}{N_h-\eta H}-\mu_v S_v,\\
E_v^\prime &=& \beta_{vh} S_v \frac{(I_h+(1-\eta)H)}{N_h-\eta H}-(\mu_v+\alpha_v)E_v,\\
I_v^\prime &=& \alpha_v E_v - \mu_v I_v,\\
\end{array}
\end{equation}

\noindent where $N_h=S_h+E_h+I_h+H+R$ and $N_v=S_v+E_v+I_v$ and each variable is described in Table \ref{tb:vars} and their respective parameters in Table \ref{tb:pars}. The model diagram is presented in Figure \ref{fig:model}.

\begin{table}[!htb]
    \caption{Model variables.}
    \label{tb:vars}
      \centering
        \begin{tabular}{|c|c|} \hline
            State Variable & Description \\ \hline 
            $S_h$   & Susceptible individuals\\ 
            $E_h$   & Exposed individuals (infected but not infectious)\\ 
            $I_h$    & Infected individuals\\ 
            $H$       & Hospitalized individuals\\ 
            $R$   & Recovered individuals\\ 
            $S_v$   & Susceptible vectors\\ 
            $E_v$   & Exposed vectors (infected but not infectious)\\ 
            $I_v$    & Infected vectors\\ \hline
       \end{tabular}
\end{table}
\begin{table}[!htb]
    \caption{Model parameters and description.}
    \label{tb:pars}
      \centering
        \begin{tabular}{|c|c|}\hline
            Parameter & Description \\ \hline 
            $\beta_{hv}$ & Transmission rate (host-vector)\\
            $\beta_{vh}$ & Transmission rate (vector-host)\\
            $\alpha_h$ & Per capita exposed rate of humans\\
            $\alpha_v$ & Per capita exposed rate of vectors\\
            $\delta$      & Per capita hospitalization rate after infection (undiagnosed)\\
            $\gamma$  & Per capita recovery rate of humans\\
            $p$             & Proportion of individuals being hospitalized (diagnosed)\\
            $\eta$         & Effectiveness of hospitalization (dimensionless)\\
            $\mu_h$     & Per capita mortality rate of humans\\
            $\mu_v$     & Per capita mortality rate of vectors\\ \hline
        \end{tabular}
\end{table}

\begin{figure}[h]
\centering
\includegraphics[width=0.8\textwidth]{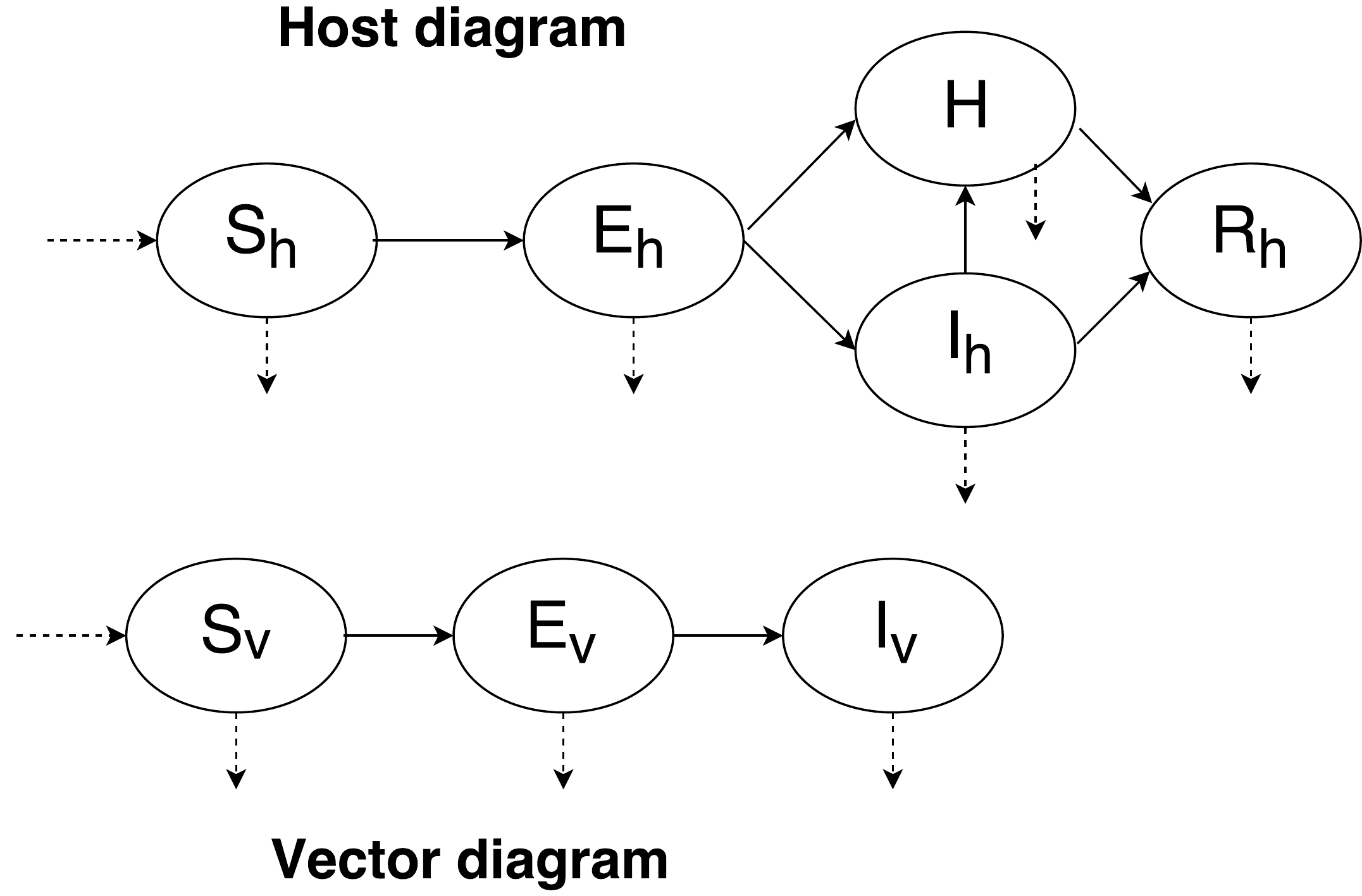}
\caption{Model diagram.}
\label{fig:model}
\end{figure}

We assume the host and mosquito remain constant in time, moreover we can have the following re-scaled variables, $s_h=\frac{S_h}{N_h}$, $e_h=\frac{E_h}{N_h}$, $i_h=\frac{I_h}{N_h}$, $\hslash=\frac{H}{N_h}$, $r_h=\frac{R}{N_h}$, $s_v=\frac{S_v}{N_v}$, $e_v=\frac{E_v}{N_v}$, $i_v=\frac{I_v}{N_v}$. Hence, the re-scaled system becomes:

\begin{equation}
\label{eq:resys}
\begin{array}{rcl}
s_h^\prime &=& \mu_h - \beta_{hv} s_h i_v - \mu_h s_h,\\
e_h^\prime &=&  \beta_{hv} s_h i_v - (\mu_h+\alpha_h)e_h,\\
i_h^\prime &=& (1-p)\alpha_h e_h - (\mu_h+\gamma+\delta)i_h,\\ 
\hslash^\prime &=& p\alpha_h e_h + \delta i_h - (\mu_h+\gamma)\hslash,\\
r^\prime &=& \gamma i_h + \gamma \hslash - \mu_h r,\\
s_v^\prime &=& \mu_v - \beta_{vh} s_v \frac{(i_h+(1-\eta)h)}{1-\eta h}-\mu_v s_v,\\
e_v^\prime &=& \beta_{vh} s_v \frac{(i_h+(1-\eta)h)}{1-\eta h}-(\mu_v+\alpha_v)e_v,\\
i_v^\prime &=& \alpha_v e_v - \mu_v i_v,\\
\end{array}
\end{equation}

\noindent where $s_h+e_h+i_h+\hslash+r=1$ and $s_v+e_v+i_v=1$. 


\begin{theorem} \label{eq:thm1}
The closed set $\Omega = \{(s_h,e_h,i_h,\hslash,r,s_v,e_v,i_v) \in \Re^{8}_{+} : 0 < s_h+e_h+i_h+\hslash+r \le 1, \;\; 0 < s_v+e_v+i_v \le 1\}$ is positively invariant for model (\ref{eq:resys}) and is absorbing.
\end{theorem}
\begin{proof}
Let $x_0\in\Omega$ be the initial state of System \ref{eq:resys}. To prove this theorem we will show that if $y=0$, then $y'\geq 0$, where $y$ is any variable of the model. Assume first that $y=0$, then notice that, from the model we have that $y' = f(x)-yg(x)$, where $x$ is the state of the system and $f,g$ are non-negative functions. Then it happens that $y'=f(x)\geq 0$ if $y=0$, therefore $y\geq 0$.

Now, since $s_h+e_h+i_h+\hslash+r=1$ and $s_v+e_v+i_v=1$ and from the previous step all variables are non-negative, then all variables $y$ satisfy that $y\leq 1$.

Therefore $\Omega$ is a positively invariant set for the model.
\end{proof}

\begin{theorem} \label{eq:thm2}
The system (\ref{eq:resys}) has exactly one equilibrium point when there is no disease in $\Omega \in \Re^{8}_{+}.$
\end{theorem}
\begin{proof}
The local equilibria are the vectors $(s_h^*,e_h^*,i_h^*,\hslash^*,r^*,s_v^*,e_v^*,i_v^*)$ such that all of the derivatives in System \ref{eq:resys} are equal to $0$. Some simple algebra leads us to the following relation between exposed individuals on each population:
$$e_h^* = \frac{\beta_h\mu_h\alpha_v}{\mu_h+\alpha_h}\frac{e_v^*}{\mu_h\mu_v+\beta_h\alpha_v e_v^*},$$

where, if
$$\hat{H} = \frac{i_h^*(1-\eta)\hslash^*}{1-\eta\hslash^*},$$

then
$$\frac{e_v^*}{\mu_h\mu_v+\beta_h\alpha_v e_v^*}, = \frac{\beta_v\mu_v}{\mu_v+\alpha}\frac{\hat{H}}{\mu_v+\hat{H}\beta_v}.$$

To turn this system into a single variable equation, notice that $i_h^*$ and $\hslash^*$ satisfy the following relationships:
\begin{align*}
\hslash^* &= \frac{\phi\alpha_he_h^*+\delta i_h^*}{\mu_h+\gamma},\\
i_h^* &= \frac{(1-\phi)\alpha_he_h^*}{\mu_h+\gamma+\delta}.
\end{align*}

By making those substitutions on $\hat{H}$ and simplifying, we get that:
$$\frac{\hat{H}}{\mu_v+\hat{H}\beta_v} = \frac{e_h^*\Delta_i}{M+e_h^*(\Delta_i-\Delta_v)},$$

where:
\begin{align*}
\Delta_i &= (1-\phi)\alpha_h(\mu_h+\gamma)+(1-\eta)\phi\alpha_h(\mu_h+\gamma+\delta)+(1-\eta)\delta(1-\phi)\alpha_h,\\
\Delta_v &= \eta\mu_v\alpha_h(\mu_h+\gamma+\delta)+\eta\mu_v\delta(1-\phi)\alpha_h,\\
M &= \mu_v(\mu_h+\gamma)(\mu_h+\gamma+\delta).
\end{align*}

If we let
\begin{align*}
\Gamma_h &= \frac{\beta_h\mu_h\alpha_v}{\mu_h+\alpha_h},\\
\Gamma_v &= \frac{\beta_v\mu_v}{\mu_v+\alpha},
\end{align*}

then we get the following equation:
$$e_h^* = \frac{\Gamma_h\Gamma_v\Delta_i e_h^*}{M+e_h^*(\Delta_i-\Delta_v)},$$

which has the following two solutions:

\begin{equation}\label{equilibria}
\begin{split}
e_h^* =& 0,\\
e_h^* =& \frac{\Gamma_h\Gamma_v\Delta_i-M}{\Delta_i-\Delta_v},
\end{split}
\end{equation}

which correspond to the disease-free and endemic equilibrium points, respectively. If we let $e_h^*$ as any of those values, by simplifying the other equations we get that:
\begin{align*}
s_h^* &= \frac{1}{k_se_h^*+1},\\
i_h^* &= k_ie_h^*,\\
\hslash^* &= k_\hslash e_h^*,\\
r^* &= k_re_h^*,\\
s_v^* &= \frac{1}{1+\kappa_se_h^*},\\
e_v^* &= \kappa_ee_h^*,\\
i_v^* &= \kappa_ie_h^*.
\end{align*}

Where the $k$'s and $\kappa$'s are parameters. The case $e_h^* = 0$ shows the existence of the disease-free equilibrium point $(1,0,0,0,0,1,0,0)$.
\end{proof}

\section{Model results} \label{sec:results}
In this section we perform the following analyses for System \ref{eq:resys}: $\mathcal{R}_0$ calculation, global equilibria of disease free equilibrium, sensitivity analysis, and some numerical simulations. 

\subsection{Basic reproductive number}
We compute the {\it basic reproductive number}, $\mathcal{R}_0$, using the next generation operator~\cite{diekmann1990}. We compute the $F$ matrix where the entries include the transmission terms for both, the host and vector populations,

\[F=
\begin{bmatrix}    
    0  & 0 & 0 & 0 & \beta_{hv} \\
    0  & 0 & \beta_{vh}  & \beta_{vh} (1-\eta) & 0 \\
    0  & 0 & 0 & 0 & 0 \\
    0  & 0 & 0 & 0 & 0 \\
    0  & 0 & 0 & 0 & 0 
\end{bmatrix}
\]

and the $V$ matrix that includes the infectious periods of both, the host and vector populations.

\[V=
\begin{bmatrix}
    (\mu_h+\alpha_h)  & 0 & 0 & 0 & 0 \\
    0  & (\mu_v+\alpha_v) & 0 & 0 & 0 \\
    -\alpha_h(1-p)  & 0 & (\mu_h+\gamma+\delta) & 0 & 0 \\
    -p \alpha_h  & 0 & -\delta & (\mu_h+\gamma) & 0 \\
    0  & \alpha_v & 0 & 0 & \mu_v 
\end{bmatrix}
\]

We then find $V^{-1}$ and compute $FV^{-1}$, where $\rho$ represents the dominant eigenvalue, and hence, the {\it basic reproductive number} is given by:

$$\mathcal{R}_0=\sqrt{\frac{\beta_{hv}}{(\mu_h+\gamma+\delta)}\frac{\beta_{vh}}{\mu_v}\frac{\alpha_h}{(\mu_h+\alpha_h)}\frac{\alpha_v}{(\mu_v+\alpha_v)}\left[ \frac{\delta(1-\eta)}{(\mu_h+\gamma)}+(1-\eta p) \right]}.$$

And can be broken down and interpreted as in Table \ref{tb:r0}.

\begin{table}[!htb]
    \caption{{\it Basic reproductive number} factors.}
     \label{tb:r0}
      \centering
        \begin{tabular}{|c|c|}\hline
            Number & Description \\ \hline 
            $\beta_{hv}$ 					 	  & Transmission rate (host-vector)\\ 
            $\frac{\alpha_h}{\mu_h+\alpha_h}$ 	  & Probability an individual survives the exposed period\\ 
            $\frac{1}{\mu_h+\gamma+\delta}$  	  & Average human infectious period\\ 
            $\beta_{vh}$ 					 	  & Transmission rate (vector-host)\\ 
            $\frac{\alpha_v}{\mu_v+\alpha_v}$ 	  & Probability a vector survives the exposed period\\ 
            $\frac{1}{\mu_v}$ 					  & Average vector infectious period\\ \hline
        \end{tabular}
\end{table}

We define,
$$\mathcal{R}_u = \frac{\beta_{hv}}{(\mu_h+\gamma+\delta)}\frac{\beta_{vh}}{\mu_v}\frac{\alpha_h}{(\mu_h+\alpha_h)}\frac{\alpha_v}{(\mu_v+\alpha_v)}(1-\eta p)$$
as the contribution of individuals that are infected and undiagnosed.
And
$$\mathcal{R}_d = \frac{\beta_{hv}}{(\mu_h+\gamma+\delta)}\frac{\beta_{vh}}{\mu_v}\frac{\alpha_h}{(\mu_h+\alpha_h)}\frac{\alpha_v}{(\mu_v+\alpha_v)}\frac{\delta (1-\eta)}{(\mu_h+\gamma)}$$
as the contributions of individuals that are hospitalized and therefore diagnosed by default.

Therefore, $\mathcal{R}_0$ can be represented by the contributions of individuals that are undiagnosed and hospitalized (diagnosed), respectively. Hence,
$$\mathcal{R}_0 = \sqrt{\mathcal{R}_u+\mathcal{R}_d}.$$

\subsection{Global equilibria} \label{sec:global}
In this section we establish the global stability of the disease-free equilibrium using the methods developed in \cite{ccc2002}.
\begin{theorem}\label{thm:global}
The disease-free equilibrium of System \ref{eq:resys} is globally asymptotically stable if $\mathcal{R}_0<1$.
\end{theorem}
\begin{proof}
Let $X\in\mathbb{R}^n$ be the uninfected individuals and $Z\in\mathbb{R}^m$ the infected individuals in the system such that the System \ref{eq:resys} is rewritten as:

\begin{equation}\label{ccc}
\begin{split}
    \frac{dX}{dt} =& F(X,Z),\\
    \frac{dZ}{dt} =& G(X,Z).
\end{split}
\end{equation}

Then, if the three following conditions are met:

\begin{enumerate}
    \item $\mathcal{R}_0<1$
    \item For $\frac{dX}{dt} = F(X,0)$, the disease-free equilibrium $X^*$ is globally asymptotically stable.
    \item $G(X,Z) = AZ-\hat{G}(X,Z)$, where $A = G_Z(X^*,0)$ and $\hat{G}(X,Z)\geq 0$ for all $(X,Z)$ where the model makes sense.
\end{enumerate}

Then the disease-free equilibrium is globally asymptotically stable.

In this section we will prove that conditions (2) and (3) are met by our model. First, consider that $X = (s_h,r,s_v)$ and $Z = (e_h,i_h,\hslash,e_v,i_v)$, then:

\begin{align*}
    F(X,Z) &= \left(\begin{array}{c}
    \mu_h-\beta_{hv} s_h i_v - \mu_h s_h\\
    \gamma i_h+\gamma \hslash - \mu_h r\\
    \mu_v - \beta_{vh} s_v \frac{i_h+(1-\eta)\hslash}{1-\eta \hslash} - \mu_v s_v
\end{array}\right),\\
G(X,Z) &= \left(\begin{array}{c}
    \beta_{hv} s_h i_v - (\mu_h+\alpha_h) e_h\\
    (1-p)\alpha_h e_h-(\mu_h+\gamma+\delta) i_h\\
    p \alpha_h e_h + \delta i_h-(\mu_h+\gamma) \hslash\\
    \beta_{vh} s_v \frac{i_h+(1-\eta)\hslash}{1-\eta \hslash} - (\mu_v+\alpha_v) e_v\\
    \alpha_v e_v - \mu_v i_v
\end{array}\right).
\end{align*}

For verifying condition (2), note that:

$$F(X,0) = \left(\begin{array}{c}
    \mu_h-\mu_h s_h\\
    -\mu_h r\\
    \mu_v-\mu_v s_v
\end{array}\right).$$

In this case, the equation:

$$\frac{dX}{dt} = F(X,0).$$

Has the solution:

$$X(t) = \left(\begin{array}{c}
    1+e^{-\mu_h}\\
    e^{-\mu_h}\\
    1+e^{-\mu_v}
\end{array}\right)$$

which satisfies that

$$\lim_{t\rightarrow\infty}X(t) = (1,0,1) = X^*.$$

Therefore the disease-free equilibrium $X^*$ is globally asymptotically stable and condition (2) holds. For condition (3), consider the following matrix:

$$A = \left(\begin{array}{ccccc}
    -(\mu_h+\alpha_h) & 0 & 0 & 0 & \beta_{hv}\\
    (1-p)\alpha_h & -(\mu_h+\gamma+\delta) & 0 & 0 & 0\\
    p \alpha_h & \delta & -(\mu_h+\gamma) & 0 & 0\\
    0 & \beta_{vh} & \beta_{vh}(1-\eta) & -(\mu_v+\alpha_v) & 0\\
    0 & 0 & 0 & \alpha_v & -\mu_v
\end{array}\right).$$

Let:

$$\hat{G}(X,Z) = \left(\begin{array}{c}
    \beta_{hv} i_v - \beta_{hv} s_h i_v\\
    0\\
    0\\
    \beta_{vh} (i_h+(1-\eta)\hslash) - \beta_{vh} s_v \frac{i_h+(1-\eta)\hslash}{1-\eta \hslash}\\
    0
\end{array}\right).$$

Note that $s_h\leq 1$ and $\frac{s_h}{1-\eta \hslash}\leq 1$, then $\hat{G}(X,Z)\geq 0$ for all $(X,Z)\in\Omega$, where $\Omega$ is given in Theorem \eqref{eq:thm1}. Then $G(X,Z) = AZ-\hat{G}(X,Z)$ and condition (3) follows.
\end{proof}

\subsection{Sensitivity analysis}
For analyzing the sensitivity of our model, we compute the sensitivity indices of the parameters with respect to $\mathcal{R}_0$ as described in \cite{chitnis2008}. These indices correspond to the partial derivatives of $\mathcal{R}_0$ with respect to each parameter evaluated on the baseline values found in Table \ref{tb:param}.

\begin{table}[!htb]  
    \caption{Parameters for dengue with baseline values, range and references. Unit of time is days.}
     \centering
      \begin{tabular}{|c|c|c|c|} \hline
	  Parameter  & Baseline & Range  & Reference \\ \hline 
	  $\beta_{hv}$  & $0.33$ & $0.10-0.75$ & \cite{manore2014}  \\ 
	  $\beta_{vh}$  & $0.33$ & $0.10-0.75$ & \cite{manore2014} \\ 
	  $1/\alpha_h$ & $5$ & $4-7$  & \cite{manore2014} \\ 
	  $1/\alpha_v$ & $10$ & $7-14$  & \cite{manore2014} \\ 
	  $\delta$     & $0.20$ & $0.10-5$ & \cite{sanchez2018} \\ 
	  $1/\gamma$ & $6$ & $4-12$ & \cite{manore2014} \\ 
	  $p$           & $0.20$  & $0-1$ & estimated \\ 
	  $\eta$        & $0.80$ & $0-1$ & estimated \\ 
	  $1/\mu_h$    & $70$ & $68-76$ & \cite{chitnis2008} \\ 
	  $1/\mu_v$    & $14$ & $8-42$ & \cite{manore2014} \\ \hline
      \end{tabular}
      \label{tb:param}
\end{table}

\begin{table}[htb!]
\centering
\caption{Sensitivity indices for $\mathcal{R}_0$.}
\label{tb:sensi}
\begin{tabular}{|c|c|}\hline
Parameter & Sensitivity index \\
\hline
$\alpha_h$ & $+0.2534$  \\
$\alpha_v$ & $+3.1677$  \\
$\beta_{hv}$  & $+2.3038$  \\
$\beta_{vh}$  & $+2.3038$  \\
$\delta$    & $-1.2037$  \\
$\eta$      & $-0.9352$  \\
$\gamma$    & $-2.8710$  \\
$\mu_h$    & $-6.4188$  \\
$\mu_v$    & $-15.0783$ \\
$p$         & $-0.5732$\\ \hline
\end{tabular}
\end{table}

These indices give us an insight on which parameters affect in a more significant manner the value of $\mathcal{R}_0$. Notice the high negative value of the sensitivity index of the mortality rate of the vector $\mu_v$, which is biologically explained by the fact that as the rate in which infected vectors are replaced by susceptible vectors grows, then the incidence of infected hosts is reduced since there are less infected vectors.

Another important mention correspond to the indices of the hospitalization rate after infection, proportion of individuals being hospitalized, and effectiveness of hospitalization ($\delta, p,\eta$ respectively). Although the values of their indices are relatively low compared to other parameters with negative indices, the parameters themselves can be increased in a more reliable manner. This can be done through educating the population on assisting to medical centers in a more frequent manner (for increasing $p$ and $\delta$) and improving the sanitary conditions of hospitals (for increasing $\eta$). The values of these indices, therefore let us understand how much can the $\mathcal{R}_0$ value be improved by increasing these parameters, and therefore reducing the population infected by dengue.

Another tool to understand the sensibility of the $\mathcal{R}_0$ value is through the global uncertainty quantification described in \cite{manore2014}. This quantification consists in developing an empirical probability distribution for the $\mathcal{R}_0$ value by assuming the parameters follow an uniform distribution in their ranges displayed in Table \ref{tb:param} and then performing random samplings of those parameters and plugging them in the value of $\mathcal{R}_0$. The probability distribution of $\mathcal{R}_0$ obtained after 100,000 samples is displayed in Figure \ref{fig:density}.

\begin{figure}[htb!]
\centering
\includegraphics[width=0.6\textwidth]{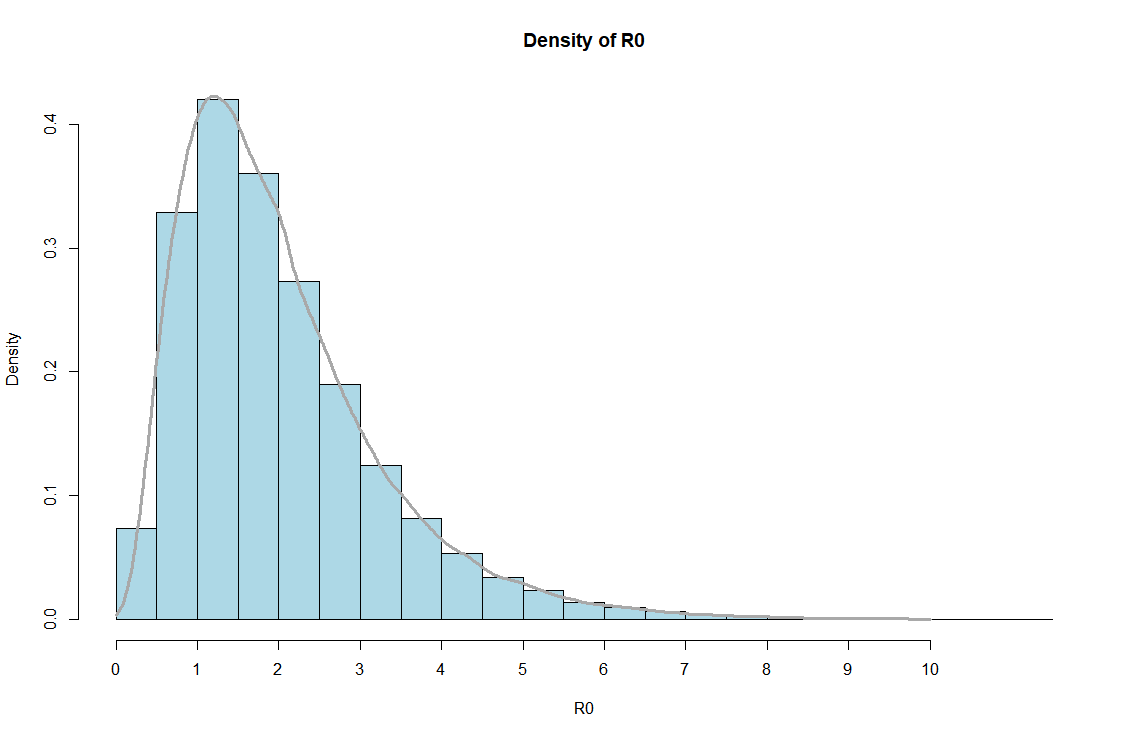}
\caption{Probability distribution of $\mathcal{R}_0$ after 100,000 samples.}
\label{fig:density}
\end{figure}

As suggested by Figure \ref{fig:density}, the most possible range for the $\mathcal{R}_0$ value lies between $0.5$ and $2$, which implies that in the case of $\mathcal{R}_0>1$, it is likely that by performing tweaks in the parameters (in a real context, that is promoting policies that alter in a real population the values of the parameters of the model) we could reduce the value of $\mathcal{R}_0$ closer to $1$ and thus significantly reduce the incidence of dengue in a human population.

\subsection{Numerical scenarios for $\eta$ and $p$} \label{sec:num}
We explored numerical experiments in attempting to find the optimal effectiveness of hospitalization of individuals. However, this is highly correlated with the number of individuals who are hospitalized due to more acute symptoms and therefore diagnosed. We can assume that the hospitalization of individuals is for the most part effective. More specifically, in Costa Rica most hospitals have the adequate equipment and staff to attend these cases.

We explored numerical experiments in attempting to find the optimal effectiveness of hospitalization of individuals. However, this is highly correlated with the number of individuals who are hospitalized. In Costa Rica, depending on the clinical manifestations and different social determinants, patients are usually sent home with basic clinical symptomatic care, recommendations, and schedule appointments in their local health care establishments to monitor evolution. Patients can also be refereed for in-hospital management or require emergency treatment and urgent referral \cite{world2012dengue,guiaccss}. The average of hospital stay among those that do require so, ranges between 3 and 4 days \cite{mallhi2016patients}, with more severe forms requiring longer stays. 

Figures \ref{fig:simIh} and \ref{fig:simH} refer to the evolution in the number of individuals infected non-hospitalized and hospitalized, respectively.

\begin{figure}[htbp]
\centering
\includegraphics[width=0.8\textwidth]{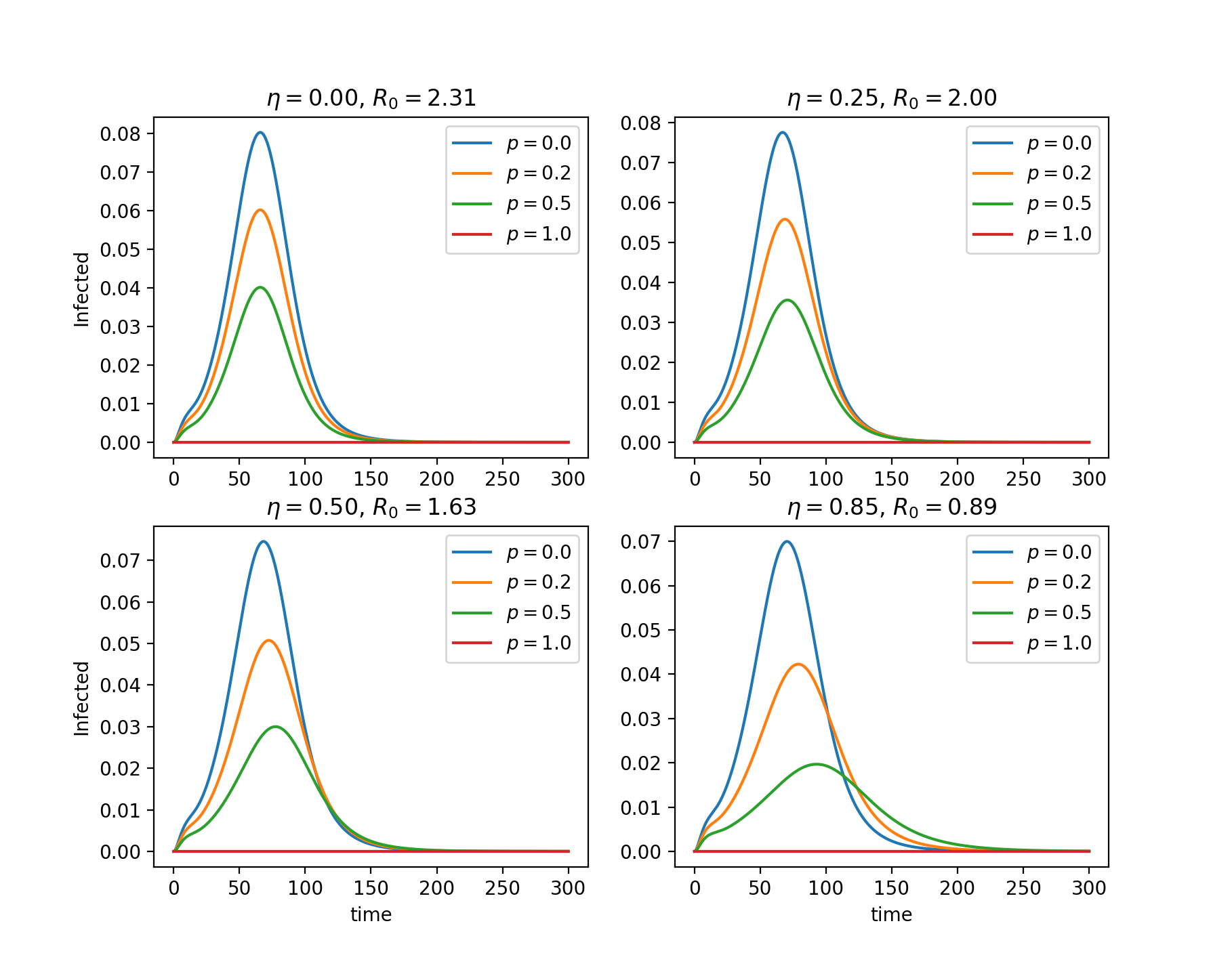}
\caption{Time series solutions of Infected (presumed undiagnosed) individuals, $i_h$. \label{fig:simIh}}
\end{figure} 

\begin{figure}[htbp]
\centering
\includegraphics[width=0.8\textwidth]{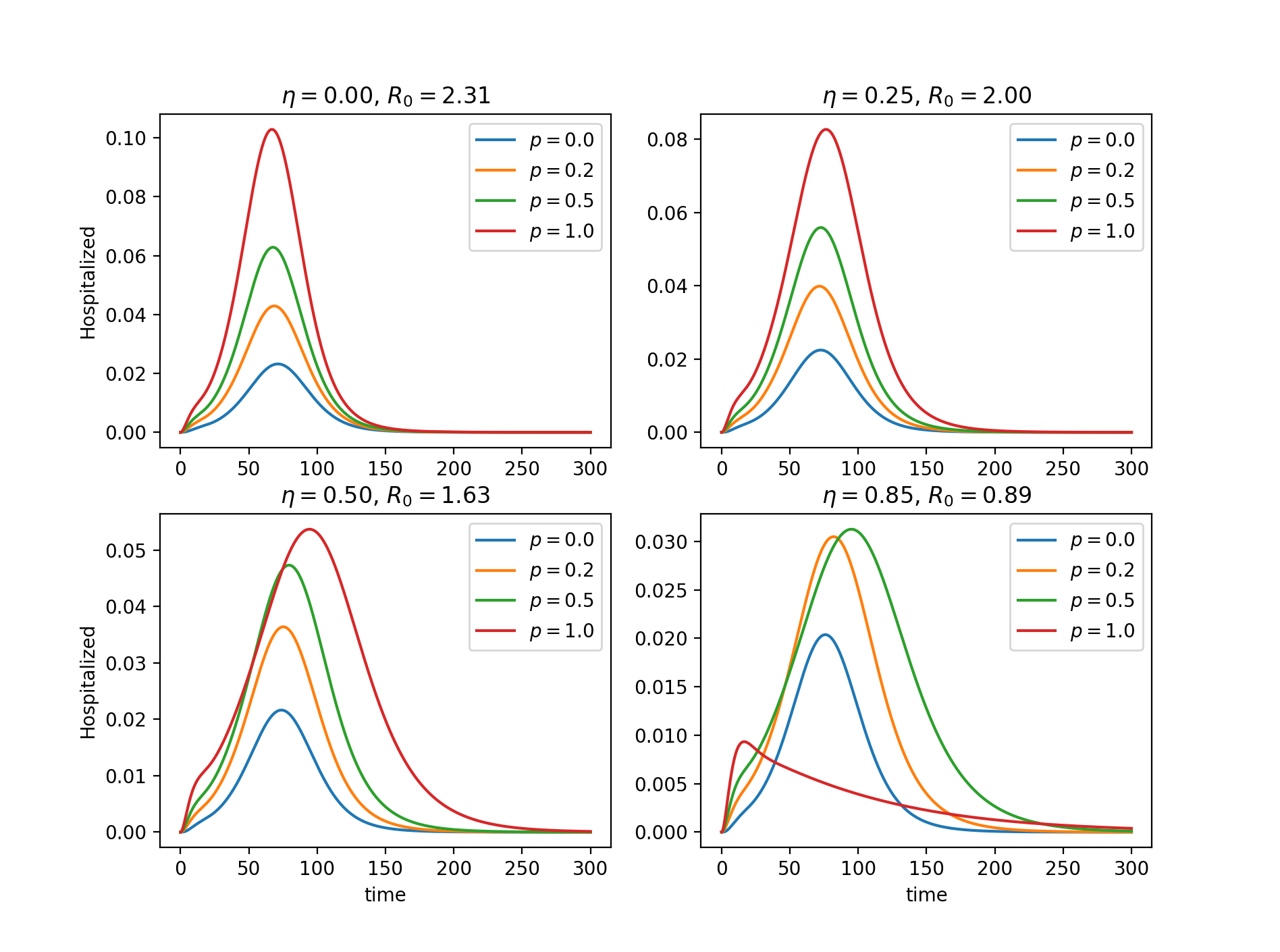}
\caption{Time series solutions of hospitalized (presumed diagnosed) individuals, $\hslash$. \label{fig:simH}}
\end{figure} 

Figure \ref{fig:params} shows how the $\mathcal{R}_0$ value behaves as $p$ and $\eta$ changes, leaving all the other parameters with their baseline values found on Table \ref{tb:param}. Notice that numerical estimations hint us that fixing the value of $p$ and increasing the value of $\eta$ lowers the $\mathcal{R}_0$ value in a faster rate than by fixing the value of $\eta$ and increasing the value of $p$. This finding relates to the difference in the sensitivity indexes found in Table \ref{tb:sensi} and suggests that, based on this model, increasing the effectiveness in hospitalization along with educating patients on preventing methods to minimize infecting others can potentially reduce dengue incidence. 

\begin{figure}[htbp]
\centering
\includegraphics[width=0.8\textwidth]{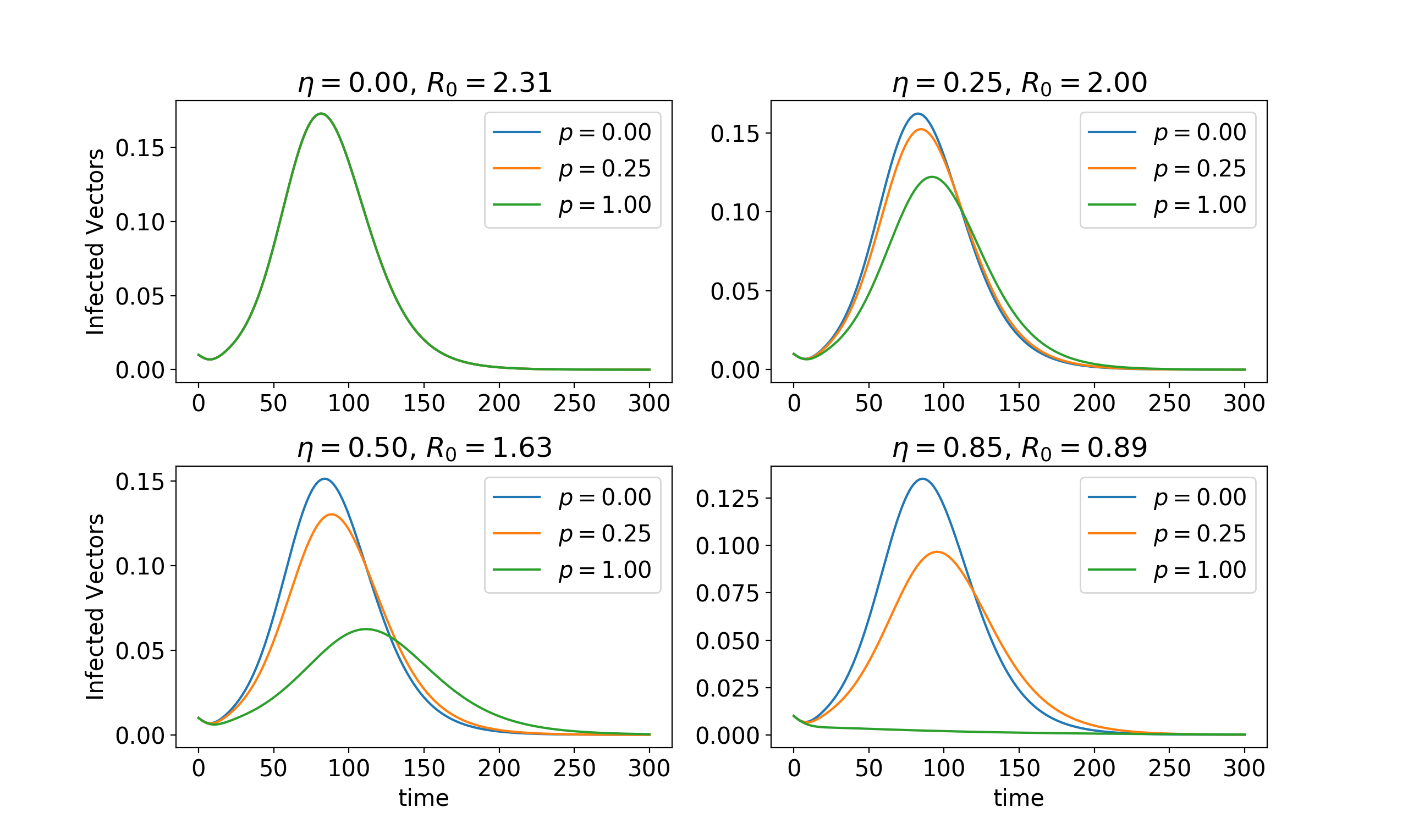}
\caption{Time series solutions of infected vectors, $I_v$. \label{fig:simIv}}
\end{figure} 

\begin{figure}[!ht]
\centering
\subfloat[$\mathcal{R}_0$ versus $p$.\label{fig:p}]{\includegraphics[width=0.5\textwidth]{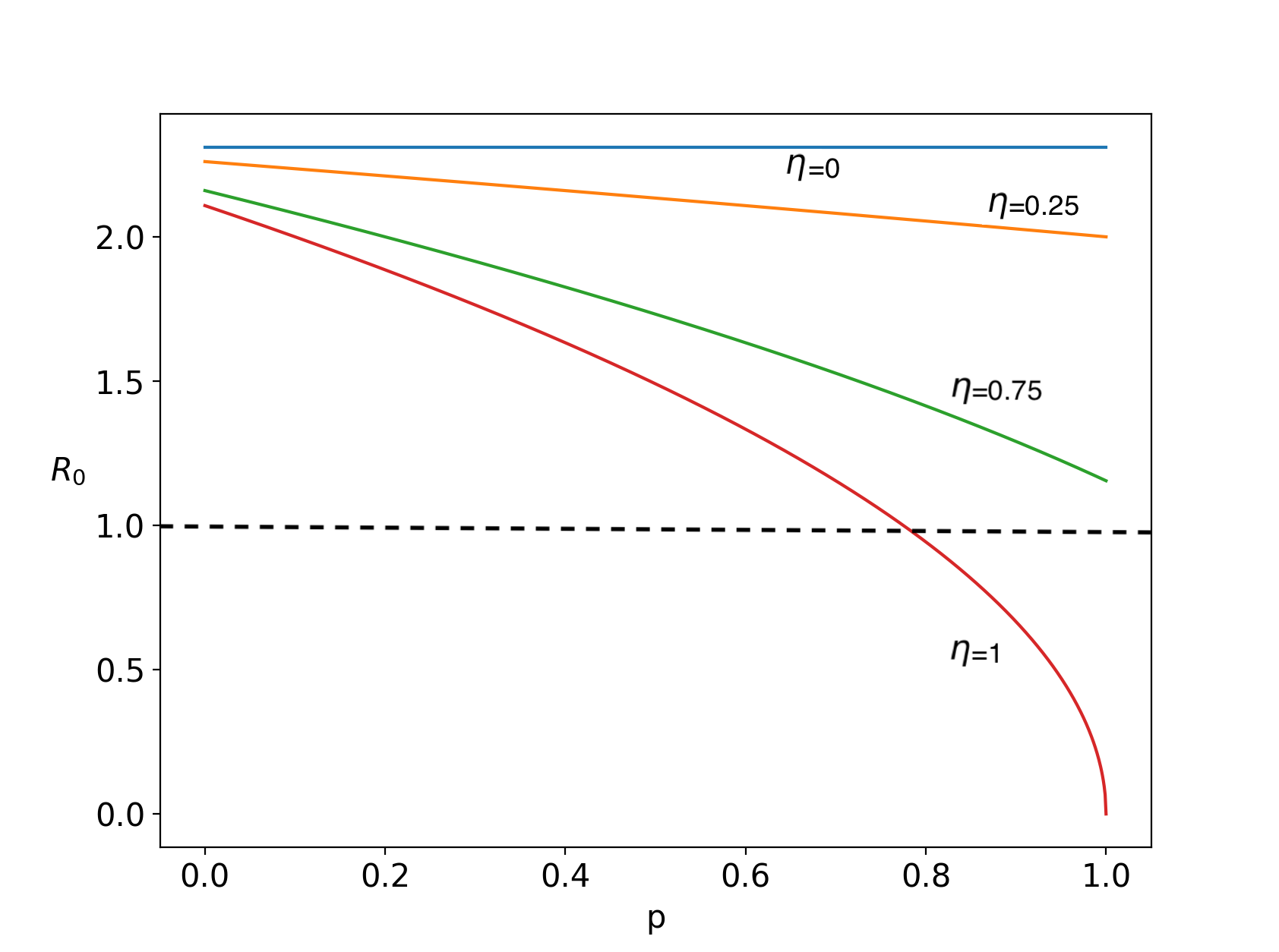}}\hfill
\subfloat[$\mathcal{R}_0$ versus $\eta$.\label{fig:eta}]{\includegraphics[width=0.5\textwidth]{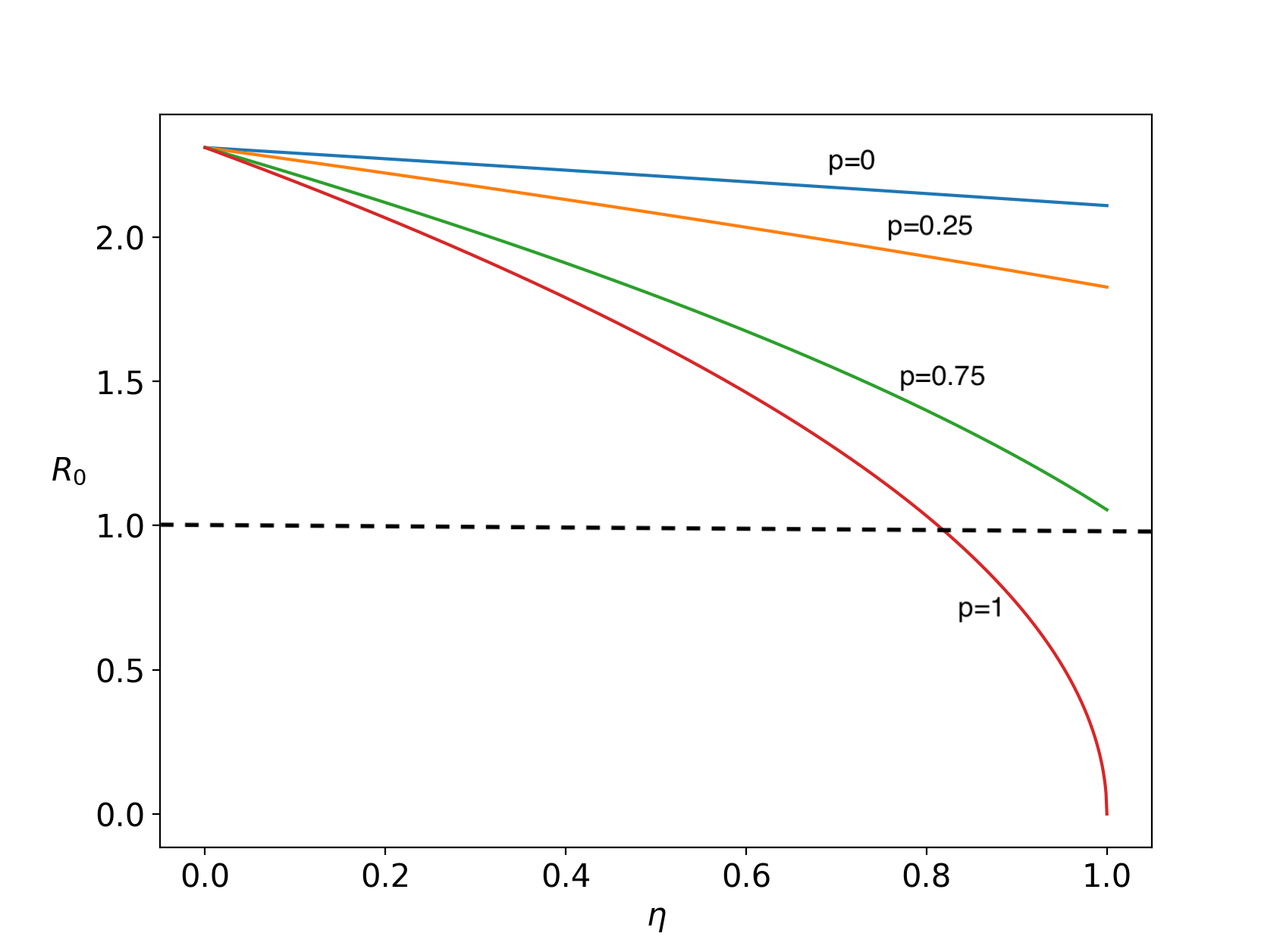}}\hfill
\caption{$\mathcal{R}_0$ versus the probability of hospitalization and effectiveness of hospitalization. \label{fig:params}}
\end{figure} 

\section{Discussion} \label{sec:disc}
In Costa Rica, as in most of the endemic countries, prevention and control strategies have focused on vector control mainly through insecticides targeting at larval or adult mosquitoes \cite{lineamientosdengue, simmons2012}. The country, also follows the recommendations made by the World Health Organization, to promote an strategic approach known as Integrated Vector Management \cite{handbook}. However, despite these efforts its proper implementation has been difficult to achieve and dengue continues to represent a mayor burden to the health care system. 

Based on the results of our model, timely and context-specific dengue contingency plans that involve providing a safer environment for those patients that stay home during their treatments, hence preventing them from propagating the virus, should be one of the priorities among public health officials. Year-round routine activities that involve more active participation from members of the community is one of the strategies that are increasingly being thought to be relevant for a successful control program \cite{romani2007achieving,al2011value,parks2004international}. Continuous capacity building in the communities will allow to reinforce local ownership, where programs adapted to the specific social, economic, environmental and geographic characteristics are a priority \cite{spiegel2005barriers}. These strategies need to go in hand with 
better coordination and communication among institutions so that successful activities of one sector will no be weakened by the lack of commitment from another. 
Also, because the clinical symptoms of dengue are so diverse and the recent emergence of other two arboviruses, each one with similar symptoms but different clinical outcomes, accurate clinical diagnosis is challenging. As such, training of health professional in diagnosis and management in conjunction with laboratory and epidemiological surveillance, is essential \cite{runge2016dengue}. Early accurate notifications of DENV infections will allow health officials to initialized promptly and targeted responses, and continues to highlight the importance and urgent need for the development of specific and sensitive point-of-care test for DENV infections \cite{liles2019evaluation}.  

The intricacies involved in the transmission dynamics of vector-borne viruses, such as the dengue virus, makes interdisciplinary collaboration essential to successfully achieve more efficient prevention and control strategies. As dengue virus continues to spread worldwide, the ever increasing need to develop and apply cost-effective, evidence-based approaches to identify and respond to potential outbreaks, has been one of the central topics from many points of view, including mathematics and medical scientists. As part of this collaboration, mathematical models have proven to be an increasingly valuable tool for the decision making process of public health authorities \cite{keeling2008}.

\section*{Acknowledgements}
We thank the Research Center in Pure and Applied Mathematics and the Mathematics Department at Universidad de Costa Rica for their support during the preparation of this manuscript. The authors gratefully acknowledge institutional support for project B8747 from an UCREA grant from the Vice Rectory for Research at Universidad de Costa Rica.

\end{document}